\documentclass[11pt]{article}
\usepackage[margin=1.3in]{geometry}

\usepackage{amsmath, amssymb, cite}
\usepackage{amsthm}
\usepackage{setspace}
\usepackage{xcolor}
\usepackage{enumitem}

\usepackage[utf8]{inputenc}
\usepackage{microtype}
\usepackage{verbatim}
\usepackage{wrapfig}
\usepackage{multirow} 
\usepackage[ruled,linesnumbered,lined,boxed,commentsnumbered]{algorithm2e}
\usepackage{changepage} 
\usepackage{graphicx}
\usepackage{url}

\theoremstyle{plain}
\newtheorem{theorem}{Theorem}

\newtheorem{lemma}[theorem]{Lemma}
\newtheorem{corollary}[theorem]{Corollary}

\date{}

\newcommand{\D}{\mathcal{D}}
\newcommand{\C}{\mathcal{C}}
\newcommand{\CD}{\tilde \D}

\newcommand{\dom}{\text{Dom}}
\newcommand{\opt}{\text{OPT}}
\newcommand{\polylog}{\text{polylog}}
\newcommand{\ID}{\text{OnlyBy}}

\newtheorem{invariant}{Invariant}

\title{Dominating Sets and Connected Dominating Sets in Dynamic Graphs}

\author{Niklas Hjuler\thanks{University of Copenhagen, Denmark, 
\textit{hjuler@di.ku.dk}}, 
Giuseppe F. Italiano\thanks{LUISS University, Rome, Italy, 
\texttt{gitaliano@luiss.it}}, 
Nikos Parotsidis\thanks{University of Rome Tor Vergata, Italy
\texttt{nikos.parotsidis@uniroma2.it}},
David Saulpic\thanks{D\'{e}partement 
d'Informatique, \'{E}cole Normale Sup\'{e}rieure, 
Paris, 
\texttt{david.saulpic@ens.fr}}
}

\date{}

\newtheoremstyle{rstate}%
{}{}%
{\itshape}{}%
{\bfseries}{}
{ }{}

\theoremstyle{rstate}

\theoremstyle{definition}
\newtheorem{definition}{Definition}

\begin{document}

\maketitle
\begin{abstract}
 In this paper we study the dynamic versions of two basic graph problems: Minimum Dominating Set and its variant
 Minimum Connected Dominating Set. For those two problems, we
 present algorithms that maintain a solution under edge insertions and edge deletions
 in time $O(\Delta\cdot \polylog~n)$ per update, where $\Delta$ is
 the maximum vertex degree in the graph.
 In both cases, we achieve an approximation ratio of $O(\log n)$, which is 
 optimal up to a constant factor (under the assumption that $P \ne NP$). Although those two problems have been widely studied in the static
 and in the distributed settings, to the best of our knowledge we are the first to present efficient algorithms in the dynamic setting.
 
 As a further application of our approach, we also present an algorithm that maintains a Minimal Dominating Set in $O(min(\Delta, \sqrt{m}\,))$ per update.
\end{abstract}

\newpage
\section{Introduction}
The study of dynamic graph algorithms is a classical area in algorithmic research and has 
been thoroughly investigated in the past decades. Maintaining a solution of a 
graph problem in the case where the underlying graph changes dynamically over time is a big challenge in the design of efficient and practical algorithms. Indeed, 
in several applications, due to the dynamic nature of today's data, it is not sufficient to compute a solution to a graph 
problem only once and for all: often, it is necessary to \emph{maintain} a solution efficiently while the input graph is undergoing a sequence of dynamic updates. 
More precisely, a \emph{dynamic graph} is a sequence of graphs $G_0, ..., G_M$ on $n$ nodes and such that $G_{i+1}$ is obtained from $G_i$ by adding or removing a single edge. The natural first barrier, in the study of dynamic algorithms, is to design algorithms that are able to maintain a solution for the problem at hand after each update faster than recomputing the solution from scratch. Many dynamic graph problems such as minimum spanning forests
(see e.g. \cite{Holm2001Poly,nanongkai2017dynamic}), shortest paths \cite{Demetrescu},
matching \cite{bernstein2019deamortization,neiman2016simple,solomon2016fully} or coloring \cite{bhattacharya2018dynamic} have been extensively studied in the literature, and very efficient algorithms are known for those problems. 
Recently, a lot of attention has been devoted to the {\sc Maximal Independent Set} problem (MIS). In this problem, one wishes to find a maximal set of vertices that do not share any edge (``maximal'' meaning that it is not possible to add any vertex without violating this property). Until very
recently, the best known update bound on the complexity to maintain a MIS was a 
simple $O(\Delta)$ algorithm, where $\Delta$ is an upper bound on the degree of vertices in the graph.
This bound was first broken by Assadi et al. \cite{assadistoc} who gave a
$O(m^{3/4})$ algorithm, then by Gupta and Khan \cite{gupta2018simple} improved 
the update bound to $O(m^{2/3})$. Very recently, using randomization, Assadi et al. \cite{assadi2019fully} presented an amortized fully-dynamic algorithm with an expected $\widetilde{O}(n^{1/2})$-time bound per update.

The MIS problem is closely related to the {\sc Dominating Set} (DS) problem: given a graph $G=(V,E)$
the DS problems is to find a subset of vertices $\D\subseteq V$ such that every vertex in 
$G$ is adjacent to $\D$ (or \emph{dominated} by $\D$).
Indeed, a MIS is also a Minimal DS:
the fact that it is not possible to add a vertex without breaking the independence property implies that every vertex is adjacent to the MIS, so this must be also a DS; at the same time, 
it is not possible to remove a vertex since that vertex is no longer dominated. Thus, to find a Minimal DS one can simply find a MIS: this gives immediately a
deterministic $O(m^{2/3})$  \cite{gupta2018simple} bound and a randomized 
$\widetilde{O}(n^{1/2})$ \cite{assadi2019fully} one. However, while it is known that is 
hard to approximate {\sc Maximum Independent
	Set}\footnote{It is not possible to find a polynomial-time algorithm that finds a
	$n^{1-\epsilon}$-approximation to {\sc Maximum Independent Set} under the 
	assumption NP $\ne$  ZPP}
within a factor $n^{1-\epsilon}$ for every $\epsilon > 0$\cite{haastad1999clique}, 
a simple greedy approach achieves a $O(\log n)$-approximation for Minimum DS \cite{greedysetcover}. 

In recent years, there has been a lot of work on designing dynamic graph algorithms for maintaining approximate solutions to several problems. 
A notable example is matching, where for different approximations there exist 
different algorithms (see e.g.,  \cite{bernstein2019deamortization,bernstein2016faster,neiman2016simple, gupta2013fully, bhattacharya2018deterministic,solomon2016fully}).
This raises the natural question on whether 
there exists a dynamic algorithm capable of 
maintaining an approximation to Minimum DS, and even better a $O(\log n)$ 
approximation. In this paper, we answer this question affirmatively by presenting 
an algorithm that achieves a $O(\log n)$ 
approximation, with a complexity matching the long standing $O(\Delta)$ bound for MIS.
Moreover, if one is interested in finding a DS faster, we present a very simple 
\emph{deterministic} $O(m^{1/2})$ algorithm to compute a Minimal DS, improving the
$O(m^{2/3})$ bound coming from MIS.
We believe these are important steps towards understanding the complexity of the problem. 
Those two results are stated below.

\begin{theorem}\label{thm:minimum}
	Starting from a graph with $n$ vertices, a $O(\log n)$ approximation of 
	{\sc Minimum Dominating Set} can be maintained over any sequence of $\Omega(n)$ edge 
	insertions and deletions in $O(\Delta \log n)$ amortized time per update, where $\Delta$ 
	is the maximum degree of the graph over the sequence of updates.
\end{theorem}

\begin{theorem}\label{thm:minimal}
	Starting from a graph with $n$ (fixed) vertices, a  
	{\sc Minimal Dominating Set} can be deterministically maintained over any sequence of edge 
	insertions and deletions in $O(\sqrt{m})$ amortized time per update, where $m$ is 
	an upper bound on the number of edges in the graph.
\end{theorem}

We also study the {\sc Minimum Connected 
	Dominating Set} problem (MCDS), which adds the constraint that the 
graph induced by the DS $\D$ must be connected. This problem was first introduced 
by Sampathkumar and Walikar \cite{sampathkumar1979connected} and arises in several 
applications. The most noteworthy is its use as a \emph{backbone} in routing 
protocols: it allows to limit the number of packet transmissions, by sending 
packets only along the backbone rather than throughout the whole network. 
Du and Wan's book \cite{duwan} summarizes the knowledge about MCDS. A special class 
of graphs is geometric graphs, where vertices are points 
in the plane, and two vertices are adjacent if they fall within a certain range 
(say, their distance is at most 1). This can model wifi transmissions, 
and the dynamic MCDS had been studied in this setting: a polynomial-time approximation
scheme is known \cite{cheng2003polynomial}, and Guibas et al. 
\cite{guibas2013connected} show how to maintain a constant-factor approximation 
with polylogarithmic update time. While geometric graphs model 
problems linked to wifi transmissions, the general graph setting can be also seen 
as a model for wired networks. However, no work about dynamic MCDS is known in 
this setting: the static case is well studied, with a greedy algorithm developed by
Guha and Keller \cite{guha1998approximation} that achieves an approximation factor 
$O(\ln \Delta)$. They also show a lower bound matching their complexity, 
together with their approximation factor. MCDS had also been thoroughly studied in the distributed
setting (see e.g. a heuristic to find a Minimal CDS in \cite{butenko2004new},
another one that sends $O(\Delta n)$ messages and has a time complexity at 
each vertex $O(\Delta^2)$ \cite{wu1999calculating} or a $3\log n$ approximation that 
runs in $O(\gamma)$ rounds where $\gamma$ is the size of the CDS found, 
with time complexity $O(\gamma \Delta^2 + n)$ and message complexity 
$O(n\Delta \gamma + m + n \log n)$ \cite{bevan1997routing}). Despite all this work, 
no results are known in the dynamic graph setting. As another application of our approach, 
we contribute to filling this gap in the research line of MCDS.
In particular, in this paper we show  how our algorithm for Minimum DS can be adapted in a non-trivial way
to maintain a $O(\log n)$ approximation of the MCDS in general dynamic graphs.

\begin{theorem}\label{thm:connected}
	Starting from a graph with $n$ vertices, a $O(\log n)$ approximation of 
	{\sc Minimum Connected Dominating Set} can be maintained over any sequence of $\Omega(n)$ edge 
	insertions and deletions in $\widetilde O(\Delta)$ amortized time per update.
\end{theorem}

We further show how to maintain independently a Dominating Set $\D$ and a set 
of vertices $\C$ such that the induced subgraphs on the vertices $\C \cup \D$ is connected. The set $\C$ has the 
additional property that $|\C| \leq 2 |\D|$, such that $|\C \cup \D| = O(|\D|)$. 
If $\D$ is a $\alpha$-approximation of Minimum DS, this gives a 
$O(\alpha)$ approximation for MCDS.

\paragraph*{Further Related Work.}
It is well known that finding a Minimum DS is NP-hard \cite{GareyJ79}. It is 
therefore natural to look for approximation algorithms for this problem. Unfortunately, 
it is also NP-hard to find a 
$c \log n$ approximation, for any $0 < c < 1$ \cite{feige1998threshold}. 
This bound is tight, since there is a simple greedy algorithm matching this bound \cite{greedysetcover}. 
Minimum DS had been studied extensively in distributed computing: an algorithm which runs in 
$O(\log n \log \Delta)$  rounds finds a  $O(\log n)$ approximation with high 
probability \cite{jia2002efficient} and an algorithm with constant number of rounds 
achieves a non-trivial approximation\cite{kuhn2005constant}. 

The DS problem is closely related to the {\sc Set Cover} problem: the two problems are 
equivalents under L-reduction \cite{kann1992approximability}. However, {\sc Set 
	Cover} was studied in the dynamic setting \cite{gupta2017online, addanki2018fully}, 
but with different kinds of updates: instead of
edges being inserted or deleted (which would represent new elements in the sets according
to the L-reduction), new elements are being added to the cover (which would be new vertices in DS).

\medskip

\noindent{\bf Outline.}
The rest of the paper is organized as follows. First, we present an algorithm 
for Minimum DS, which will be used later on also for MCDS: we start by a $\widetilde O(n)$ 
algorithm, and then show how to overcome its bottleneck in order to 
achieve a $\widetilde O(\Delta)$ complexity. Finally, we present our 
$O(\sqrt m)$ algorithm for Minimal DS.

\section{A $O(\log n)$ approximation of {\sc Minimum Dominating Set} in $O(\Delta \log 
	n)$ time per update}\label{sec:minimum}

This section aims at proving Theorem~\ref{thm:minimum}.
Following a reduction from Set Cover, minimum DS is NP-hard to 
approximate within a factor $\log n$ \cite{feige1998threshold}. Here we present a matching upper bound 
(up to a constant factor), in the dynamic setting. 
Our algorithm relies heavily on the clever set cover algorithm by Gupta et al.  \cite{gupta2017online}. While in the static setting Set Cover is equivalent to 
minimum DS, in the dynamic setting these two problems are different. More precisely, in the dynamic Set Cover problem one is asked to cover a set of points $S$ (called the universe) with a  given family of sets $F$, while the set $S$ is changing dynamically. To draw the parallel with DS, in the latter the set $S$ is the 
set of vertices of the graph (which does not change) and for every vertex the 
set of its neighbors is in $F$. The dynamic part concerns therefore $F$, and 
not the universe $S$.

Gupta et al. present an $O(\log n)$-approximation for dynamic Set Cover problem: in what follows, we show how to adapt their algorithm to the DS case, with an update time of $O(\Delta \log n)$. 
As in \cite{gupta2017online}, the approach easily adapts to the weighted case. Unfortunately, 
this cannot be generalized to MCDS, therefore we do not consider this property 
of the algorithm. The following definitions are partly adapted from \cite{gupta2017online}.

\subsection{Preliminaries}
For a vertex $v$, let $N(v)$ be the set of its neighbors, including $v$. 
The algorithm maintains a solution $S_t$ at time $t$ such that an element of 
$S_t$ is a pair composed of
\begin{itemize}
	\item a dominant vertex $v$
	\item a set $\dom(v) \subseteq N(v)$, which are the vertices that are 
	dominated by $v$. We call $|\dom(v)|$ the \emph{cardinality} of the pair.
\end{itemize}
We call a \emph{dominating pair} an element of $S_t$. The algorithm requires that
multiple copies of a vertex can appear as the dominant vertex of a pair. 
However, each vertex is exactly in one $\dom(v)$. The solution to
the DS problem is composed of all vertices that appear as dominant vertices of a pair. 
Since each vertex is in exactly one $\dom(v)$, each vertex is dominated and 
therefore the set of dominants is a valid solution to the DS problem.

The dominating pairs are placed into \emph{levels} according to their cardinality: the level $l$ is defined by a range $R_l := [2^{l-10}, 2^l]$, and each
pair $(v, \dom(v))$ is placed at an appropriate level $l$ such that $|\dom(v)| \in R_l$. In 
that case, elements of $\dom(v)$ are said to be dominated at level $l$; we 
denote by $V_l$ the set of all vertices dominated at level $l$. 
We say that an assignment of levels is valid if it respects the constraint $|\dom(v)| \in 
R_l$.
This allows us to define the notion of \emph{Stability}:
\begin{itemize}
	\item \emph{stable solution}: A solution $S_t$ is stable if there is no 
	vertex $v$ and level $l$ such that $|N(v) \cap V_l| > 2^l$; in other words, 
	it is not possible to introduce a new vertex in the solution to dominate some 
	vertices at level $l$ such that the resulting dominating pair could be at 
	level strictly greater that $l$.
\end{itemize}

The algorithm will dynamically maintain a stable solution $S_t$, with a valid 
assignment of levels. Note that the ranges $R_l$ overlap: this gives some slack to the algorithm, which allows enough flexibility
to prevent too many changes while our algorithm maintains a valid solution.

\subsection{The algorithm}
The main part of the algorithm is the function {\sc Stabilize}, which restores the 
stability at the end of every update. The function is the following (see \cite{gupta2017online}):
\begin{quote}
	{\sc Stabilize}. As long as a vertex $v$ violates the stability condition at level $l$, do the following: Add the 
	pair $(v, N(v) \cap V_l)$ to the \emph{lowest} possible level $j$ (i.e., the 
	lowest level such that $|N(v) \cap V_l| \in R_j$); Remove the elements of $N(v) 
	\cap V_l$ from the set of their former covering pair: if it gets empty, remove the
	pair from the solution. Otherwise, if the cardinality of such a pair 
	goes below $2^{l-10}$, put it at the \emph{highest} possible level. 
\end{quote}

\smallskip
\noindent \textbf{Edge addition:} When a new edge $(u, 
v)$ is added to the graph, one just need to ensure that the solution remains stable, and thus the 
algorithm runs {\sc Stabilize}.

\medskip
\noindent \textbf{Edge deletion:} When an edge $(u, v)$ is removed from the graph, we proceed as follows.  If 
neither $u$ nor $v$ dominates the other endpoint, the solution remains 
valid and stable, and nothing needs to be done. Otherwise, assume without loss of generality that  $v$ dominates $u$. Then: 
\begin{itemize}
	\item Remove $u$ from $\dom(v)$
	\item Add the pair $(u, \dom_u = \{u\})$ to the solution with level 1
	\item Run {\sc Stabilize}
\end{itemize}

\subparagraph*{Correctness.} All the nodes of the graph are dominated at every 
time. Indeed, {\sc Stabilize} does not make any node undominated and if a vertex is 
not dominated after an edge removal, the algorithm simply adds it to the 
solution. Therefore, the solution $S_t$ maintained by the algorithm is a valid one.

\subsection{Analysis}

\noindent{\bf Approximation ratio.}
We use the following lemma by 
Gupta et 
al.~\cite{gupta2017online}
to bound the cost
of a stable solution.
\begin{lemma}[Lemma 2.1 in \cite{gupta2017online}]\label{lem:stablecost}
	The number of sets at one level in any stable solution is at most $2^{10}\cdot \opt$.
\end{lemma}

Since for every dominating pair 
$(v, \dom(v))$ we have that $1 \leq |\dom(v)| \leq n$, there are only $\log n$ levels that 
can contain a set. The total cost of a stable solution is therefore 
$O(\log n \cdot \opt)$.

\medskip

\noindent{\bf A token scheme to bound the number of updates.}
Unfortunately, the analysis of Gupta et al. cannot be 
applied directly to the case of DS, due to the different nature of the updates. However, we can build upon their analysis, as follows. 
We first bound the number of vertices that change
level, and then explain how to implement a level change so that it costs 
$O(\Delta)$. We prove the following lemma by using a token argument. 

\begin{lemma}\label{lem:lvl_change}
	After $k$ updates of the algorithm, at most $O(k\log n + n\log n)$ elements 
	have changed levels.
\end{lemma}
\begin{proof}
	We use the following token scheme, where each vertex pays one token for each level change. In the beginning, we give $2\log n$ tokens to every vertex. If a vertex is undominated after an edge removal, we give 
	$2\log n$ new tokens to this vertex. Since at most one vertex gets undominated for each edge deletion, the total number of tokens given after $k$ updates is  
	$O(k\log n + n\log n)$.  To prove the lemma, we need to show that at any time
	each vertex has always a positive 
	amount of tokens. 
	We adapt the proof of Gupta et al. to show the following invariant: 
	\begin{invariant}
		Every vertex at level $l$ has more than $2(\log n - l)$ tokens.
	\end{invariant}
	
	When a vertex is moved to a higher level, it pays one token for the cost of 
	moving. It also saves one token, and gives it to an ``emergency fund''
	of its former covering pair. Each pair has therefore a fund of tokens that 
	can be used when the pair has to be moved to a lower level.
	
	When the pair $(v,\dom(v))$ has to be moved from level $l$ to level $l-j$, it 
	means that a lot of vertices have left $\dom(v)$ and that the tokens they gave to the pair
	can be used to pay for the operation. Formally, we want to pay one token for 
	every vertex in $\dom(v)$ for its level change, but we also want to restore the 
	invariant. We need therefore $2j + 1$ tokens for each vertex of $\dom(v)$. 
	Since the pair can be moved to level $l-j$, this means that $|\dom(v)| < 2^{l-j}$. 
	Since 
	a new pair is moved to the lowest possible level, this pair could not be at level 
	$l-1$, which implies that $|\dom_{init}(v)| > 2^{l-1}$ where $\dom_{init}(v)$ is 
	the set $\dom(v)$ at the time where it was created. Moreover, each of 
	the vertices that left gave one token: the amount of tokens usable is therefore 
	bigger than $2^{l-1} - 2^{l-j}$. Thus we want to prove that 
	$2^{l-1} - 2^{l-j} \geq (2j +1)\cdot|\dom(v)|$. It is enough to have $2^{l-1} 
	- 2^{l-j} \geq 3 \cdot (2j + 1) 2^{l-j}$, i.e. to have $2^{j-1} - 1 \geq 3(2j + 
	1)$. But since the pair was moved to level $l-j$, it means that $|\dom(v)| > 
	2^{l-j-1}$ and $|\dom(v)| < 2^{l-10}$: putting these two equations together 
	gives $j > 9$, 
	which ensures that $2^{j-1} - 1 \geq 3(2j + 1)$ and concludes the proof.
\end{proof}

As the following corollary shows, we can
bound the number of changes to $\D$ to 
$O(\log n)$ amortized.
This property will be useful in Section \ref{sec:simple}.

\begin{corollary}\label{cor:num_change}
	After $k$ updates of the algorithm, at most $O(k\log n + n\log n)$ vertices can be 
	added to or removed from $\D$.
\end{corollary}
\begin{proof}
	Whenever a vertex is added to or removed from $\D$,  its level is changed. Lemma 
	\ref{lem:lvl_change} gives the corresponding bound.
\end{proof}

We now turn to the implementation of the function {\sc Stabilize}. As shown in
the next lemma, we implemented so that its cost is $O(\Delta)$ 
for each element that changes level.

\begin{lemma}\label{lem:stablecomplex}
	A stable solution can be maintained in $O(\Delta \log n)$ amortized time per update.
\end{lemma}
\begin{proof}
	For all vertices $v$ and all levels $l$, the algorithm maintains the set 
	$N(v)\cap V_l$ and its cardinality. Every time a vertex changes its level, it has to 
	inform all its neighbors: this can be done in $O(\Delta)$. When an edge $(u, v)$
	is added to or removed from the solution, the algorithm updates the sets $N(v)\cap 
	V_{l_u}$ and $N(u)\cap V_{l_v}$, where $l_u$ and $l_v$ are the levels of $u$ 
	and $v$, respectively. 
	
	During a call to {\sc Stabilize}, the algorithm maintains also a list of vertices 
	that may have to be added to restore the stability: for a vertex $v$ and 
	level $l$, every time that $N(v)\cap V_l$ changes, if the new 
	cardinality violates the stability, we add $v$ to this list in constant time. The algorithm processes the list vertex by vertex: it checks that the 
	current vertex still needs to be added to the solution, and add it if necessary.
	
	Since we pay $O(\Delta)$ per level change and there are $O(\log n)$ amortized changes, the amortized complexity of each update is $O(\Delta \log n)$. 
\end{proof}

Since a stable solution gives a $O(\log n)$ approximation to minimum DS, 
Lemmas \ref{lem:stablecost} and \ref{lem:stablecomplex} yield the proof
of Theorem~\ref{thm:minimum}: a $O(\log n)$ approximation of Minimum Dominating 
Set can be maintained in $O(\Delta \log n)$ amortized time per update.

\section{A $O(\log n)$ Approximation for {\sc Minimum Connected Dominating Set} in 
	$\widetilde O(n)$ per update}\label{sec:simple}

A possible way to compute a Connected DS is simply to find a 
DS and add a set of vertices to make it connected. Section 
\ref{sec:minimum} gives an algorithm to maintain an approximation of the 
Minimum DS: we will use it as a black box (and refer to it as the ``black box''), 
and show how to make its 
solution connected without losing the approximation guarantee. If the original
graph is not connected, the algorithm finds a CDS in every connected component:
we focus in the following on a single of these components. Let $\D$ be the DS 
maintained, and $\C$ be a set of vertices such that $\C \cup \D$ is connected 
and $\C$ is minimal for that property.
The minimality of $\C$ will ensure that $|\C| \leq 2 |\D|$: since $\D$ is a 
$O(\log n)$ approximation of MDS, this leads to a $O(\log n)$ approximation 
for MCDS. Note that the vertices of $\C$ are not used for domination: $\C 
\cup \D$ is therefore not minimal, but still an approximation of minimum.

Overall, we will apply the following charging scheme to amortize the total 
running time. The main observation is that although a lot of vertices can be deleted
to restore the minimality of $\C$, only a few can be added at every step.
We thus give enough potential to a vertex whenever it is added 
into $\C$ and whenever its neighborhood changes, so that at the time 
of its removal from $\C$ it has accumulated enough potential for scanning its entire 
neighborhood. After an edge deletion we might have to restore the 
connectivity requirement. We do that by adding at most 2 new 
vertices in $\C$: this is crucial for our amortization argument. 

\bigskip
\textbf{Outline:} The set $\C$ may have to be updated for two reasons: 
\begin{itemize}
	\item Restore the connectivity: if an edge gets deleted from the graph, or if 
	the black box removes some vertices from $\D$, it may be necessary to add some 
	vertices to $\C$ in order to restore the connectivity of $\C \cup \D$.
	\item Restore the minimality of $\C$: when an edge is added to the graph, or 
	when a vertex is added to $\C \cup \D$ (either by the black box or in order to restore 
	the connectivity), some vertices of $\C$ may become useless and therefore need 
	to be removed.
\end{itemize}
We now address those two points. All our bounds are expressed in term 
of the total number of changes in $\C \cup \D$: let therefore $k$ be this number of 
changes. We will show later that, after $t$ updates to the graph, $k = O(t\log n)$.

\bigskip
The first phase of the algorithm is to restore the connectivity. We explain in the following
how to decide which vertices should be added to $\C$ for that purpose.

\paragraph*{Restore the connectivity after an edge deletion.}
To monitor the connectivity requirement, we use the following idea. The algorithm
maintains a minimum spanning tree (MST) of the graph $G$ where a weight $1$ is assigned
to the edges between vertices in $\C \cup \D$ (called from now on $\CD$), and 
weight $m$ is assigned to all other edges. These weights ensure that, 
as long as $\CD$ is connected, the MST induces a tree on $\CD$.
When $G[\CD]$ gets disconnected by an update, the MST uses a
vertex of $V \setminus \CD$ as an internal vertex: in that case, our 
algorithm adds this vertex to $\C$, to restore the connectivity. 
We give more details in the next section.

The edge weights are updated as the graph undergoes edge insertions and deletions 
and vertices enter or leave $\CD$. The MST of the 
weighted 
version of the graph has the following properties.

\begin{itemize}
	\item If $\CD$ is a connected DS, then the MST has weight 
	$(|\CD|-1)+m\cdot|V\setminus \CD|$ (Kruskal's algorithm on this graph 
	would use $|\CD|-1$ edges of weight 1 to construct a spanning tree on 
	$\CD$, then $|V\setminus \CD|$ edges of weight $m$ to span the entire 
	graph).
	\item If $\CD$ is a DS but $G[\CD]$ is not connected, then the 
	weight of the MST has larger value.
\end{itemize}

The two properties stem from the fact that a MST can be produced by finding a
minimum spanning forest on $\CD$  and extend it 
to a MST on $V$. Kruskal's algorithm ensures that this leads to a MST. 
In the case where $\CD$ is
connected, the first step yields a tree of weight $\CD - 1$, and since the 
graph is connected the second step yields a cost $m\cdot|V\setminus \CD|$. However,
if $\CD$ is not connected, the second step adds strictly more that 
$|V\setminus \CD|$ edges, therefore yielding a cost bigger than $m\cdot(1+|V\setminus \CD|)$.
This is more than $(|\CD|-1)+m\cdot|V\setminus \CD|$, as claimed.

Furthermore, if $G[\CD]$ has two connected components $C_1,C_2$, then the 
shortest of all paths between vertices $u,v,~u\in C_1, v\in C_2$ is the minimum 
number of vertices whose insertion into $\C$ restores the connectivity 
requirement.
Note that the 
shortest of all such paths must have length at most $2$ (otherwise, there must be a vertex 
not adjacent to any vertex in $\D$, 
which contradicts the fact that $\D$ is 
a DS). 

After an edge deletion, it may happen that $\CD$ becomes disconnected and that 
the MST includes some internal vertices (at most 2, by the previous discussion) not in 
$\CD$: in that case, we add them to $\C$. This turns out to be enough to ensure 
the connectivity.

To maintain the MST of the weighted version of the input graph we use the 
$O(\log^4 n)$ update time fully-dynamic MST algorithm from \cite{Holm2001Poly}.
Since the weights of the edges incident to the vertices that enter or leave 
$\CD$ are also updated, the algorithm runs in time $\widetilde{O}(\Delta)$ 
for each change in $\CD$, i.e. in time $k\cdot \widetilde O(\Delta)$

\medskip

\noindent{\bf 
	Restore the connectivity when a vertex is deleted by the black box.}
\label{par:connec}
When a vertex $v$ is deleted from $\D$ by the black box 
DS algorithm, we 
need to be more careful: updating the edge weights  and finding the new MST may add a lot of vertices to $\C$ (as 
many as $\Delta$, one per edge of the MST incident to $v$). However, if the 
removal of $v$ disconnects $G[\CD]$, it is enough to add $v$ to $\C$ to restore 
the connectivity. If its removal does not disconnect $G[\CD]$, nothing needs to be done. It is possible to know if the graph $G[\CD]$ gets disconnected using 
the properties of the MST, by only looking at the weight of the MST. The 
complexity of this step is therefore $\widetilde O(\Delta)$, the time needed
to update the weights of the MST.

\medskip

\noindent{\bf Restore the minimality.}
The second phase of the algorithm is to restore the minimality of $\C$. We explain 
next how to find the vertices of $\C$ that need to be removed to accomplish this task.
This minimality condition is equivalent to the condition that all vertices in $\C$ 
are \emph{articulation points} in the graph induced by $\C \cup D$. 
(An articulation point is a vertex such that its removal increases the number of 
connected components.) This turns out to be useful in order to identify which vertices need to 
be removed to restore the minimality of $\C$.

To restore the connectivity requirement, new vertices were added into $\C$, and 
the black box added some vertices to $\D$: this might result in some 
vertices in $\C$ not being articulation points of $G[\CD]$ anymore. As observed 
before, these are the vertices that need to be removed. We need to identify a 
maximal set of such vertices that can be removed from $\C$ without violating the 
connectivity requirement.
To do this, 
the algorithm queries in an 
arbitrary order one-by-one all the vertices $v\in \C$ to determine whether 
$G[\CD \setminus v]$ is connected. This can be done using a data structure from
Holm et al.~\cite{Holm2001Poly} that requires $\widetilde O(1)$ per query.
Whenever the algorithm identifies a vertex such that $G[\CD \setminus v]$ is 
connected, it can safely remove it from $\C$. The complexity of this step is 
therefore $\widetilde O(n)$ to find all articulation points, and an extra
$\widetilde O(\Delta)$ for each of the vertices we remove from $\C$.

The following three lemmas conclude the proof: the first shows that 
the algorithm is correct, the second the $\widetilde O(n)$ time bound and the third the $O(\log n)$ approximation ratio.

\begin{lemma}\label{lem:slow}
	The algorithm that first restores the connectivity of $\C \cup \D$ and then the 
	minimality of $\C$ is correct: it gives a minimal set $\C$ such that 
	$\C \cup \D$ is connected.
\end{lemma}

\begin{proof}
	After restoring the connectivity requirement the algorithm maintains a 
	spanning tree of $\CD$, so $G[\CD]$ is indeed connected. 
	In the following steps, before the algorithm removes a vertex $v$ from $\C$, it 
	first verifies that $G[\CD \setminus v]$ remains connected, which guarantees that $G[\CD]$ is 
	connected at the end of the update procedure.
	Since the black box ensures that $\D$ is a DS, $\CD$ is a DS 
	too: hence at the end, $\CD$ satisfies both the domination and the connectivity 
	requirements.
	It remains to show that 
	$\C$ is minimal, i.e., that all 
	vertices in $\C$ are articulation points in $G[\CD]$. Since during the second 
	step the algorithm only removes vertices from $\C$, a vertex that was not an 
	articulation point cannot become one, and therefore the loop to find the 
	articulation points is correct.
	The set $\C$ is therefore a minimal set such that $\C \cup \D$ is connected.
\end{proof}

\begin{lemma}\label{lem:simple_complex}
	The amortized complexity of the algorithm is $\widetilde O(n)$ per update.
\end{lemma}

\begin{proof}
	The amortized cost of the black box to compute $\D$ is $\widetilde O(\Delta)$. 
	We analyze now the additional cost of maintaining $\CD$. As shown in this 
	section, the cost to add or delete a vertex from $\CD$ is $\widetilde O(\Delta)$.
	To prove the lemma, we bound the number of changes in $\CD$. For that,
	we count the number of vertices \emph{added} to $\CD$: in an amortized sense this
	bounds the number of changes too. Formally, we pay a budget $deg(v)$ when
	$v$ is added to $\CD$. Following insertions and deletions of edges adjacent to $v$,
	we update this budget (with a constant cost), so that when $v$ gets deleted from
	$\CD$ a budget equal to its degree is available to spend. 
	
	From Corollary \ref{cor:num_change}, the black box makes at most $\widetilde O(1)$ 
	changes to $\D$ per update (in an amortized sense). 
	If it removes a vertex from $\D$, we showed previously 
	that no new vertex is added to $\CD$. The number of additions to $\CD$ is therefore
	$\widetilde O(1)$. Moreover, in the case of an edge deletion, at most two 
	vertices are added to $\CD$ to maintain the connectivity. 
	Since restoring the minimality requires only to delete vertices, the total number
	of additions into $\CD$ is $\widetilde O(1)$. As the cost for any of these additions
	is $\widetilde O(\Delta)$, the total cost of this algorithm is upper bounded by
	the loop to find the articulation points, which is $\widetilde O(n)$.
\end{proof}

\begin{lemma}\label{lem:approx_cds}
	The algorithm maintains a $O(\log n)$ approximation for MCDS, i.e. $|\C \cup 
	\D| = O(\log n)\cdot \opt$
\end{lemma}
\begin{proof}
	We first prove that $|\C| \leq 2|\D|$, using the minimality of $\C$.
	Each vertex of $\C$ is there to connect some components of $\D$. Consider the 
	graph $(W, F)$ where vertices $W$ are either connected components of $\D$ or 
	vertices of $\C$, and the set $F$ of edges is constructed as follows. Start with a graph containing one vertex for each connected 
	component of $\D$, and add vertices of $\C$ one by one. When the vertex $v$ is 
	added, identify a node $u$ in $\D$ adjacent to $v$ such that adding the edge 
	$(u,v)$ to $F$ does not create a cycle: add to $F$ an edge between $v$ and the 
	node corresponding to the connected component containing $u$. It is always possible to find 
	such a vertex $u$, otherwise $v$ would not be necessary for the connectivity, which 
	would contradict the minimality of $\C$. This process gives a forest such that 
	every node of $\C$ is adjacent to a connected component of $\D$. Since $\C \cup 
	\D$ is connected, it is possible to complete $F$ to make it a tree, adding some 
	other edges.
	This tree has the two following properties.
	\begin{enumerate}
		\item The leaves are vertices that correspond to connected component of $\D$: 
		indeed, if a vertex of $\C$ was a leaf in this tree, it could be removed 
		without losing the connecting of $\C \cup \D$, which would contradict the 
		minimality of $\C$. 
		\item Any vertex of $\C$ is adjacent to a connected component of $\D$, by 
		construction of the forest.
	\end{enumerate}
	
	These properties ensure that for every subtree rooted at a vertex of $\C$, 
	there is a $\D$ vertex at distance at most 2 from the root: otherwise, the
	vertices at distance 1 from it would be from $\C$ and adjacent only to $\C$ vertices.
	Moreover, since a $\C$ vertex is not a leaf, it has necessarily some descendant and
	the reasoning applies.
	Therefore, by rooting the tree at an arbitrary vertex of $\C$, we can charge every
	$\C$ vertex to a $\D$ descendant at distance at most 2. As a $\D$ vertex can be
	charged only by an ancestor at most two levels above it, it is charged at most twice.
	This ensures that $|\C| \leq 2|\D|$.

	Moreover, since $\D$ is a $O(\log n)$ approximation of MDS, $|\D| = 
	O(\log n)\cdot \opt$. Putting things together, we have $|\C \cup \D| = |\C| + 
	|\D| = O(\log n)\cdot \opt$.
\end{proof}

Combining Lemmas \ref{lem:slow}, \ref{lem:simple_complex} and \ref{lem:approx_cds} 
proves our claim: there is a $\widetilde O(n)$ algorithm to maintain a 
$O(\log n)$ approximation of the Minimum Connected Dominating Set. 
The main bottleneck of this approach is the time spent by the algorithm in the 
second phase to query all vertices in $\C$ in order to identify the vertices that 
are no longer articulation points. In the next section we present an 
algorithm that overcomes this limitation and is able to  identify the necessary vertices 
more efficiently.

\section{A more intricate $\widetilde O(\Delta)$ algorithm to restore the 
	minimality of $\C$}

In this section we present a more sophisticated algorithm for implementing the 
phase that guarantees the minimality of the maintained connected dominating 
set. This gives a proof of Theorem~\ref{thm:connected}.
We focus on a single edge update: indeed, when a vertex is added to (or 
removed from) $\CD$, one can simply add (or remove) all its edges one by one. 
As in the analysis of the complexity in Lemma~\ref{lem:simple_complex}, the 
amortized number of changes in $\CD$ is $\widetilde O(1)$. We aim now at proving
that the time required for handling a single change is 
$\widetilde O(\Delta)$: for that, we treat edge insertions and deletions to $\CD$
one by one, and prove that any edge update can be done
in $\widetilde O(1)$, which
would prove the claimed bound. 
Our algorithm maintains another spanning forest $F$ of $G[\CD]$ (unweighted) 
using the algorithm from \cite{Holm2001Poly}.

\begin{lemma}\label{lem:oneremoval}
	The vertices of $\C$ that are not articulation points after the insertion 
	of the edge $(v, w)$ all lie on the tree path $v...w$ of $F$. Moreover, the 
	removal of any of these vertices results in the other vertices being 
	articulation points again.
\end{lemma}
\begin{proof}
	Let $G_b$ be the graph before the insertion of $(v, w)$, and $G_a$ be the one 
	after.
	Let $u$ be a vertex that is an articulation point in $G_b[\CD]$ but not in 
	$G_a[\CD]$. Suppose by contradiction that $u$ is not on the tree path 
	$v...w$: that means that $v$ and $w$ are connected in $G_b[\CD] \setminus 
	\{u\}$. Since $u$ is an articulation point in $G_b[\CD]$, $v$ is not connected 
	to some vertex $x$ in $G_b[\CD] \setminus \{u\}$. But as $v$ and $w$ are 
	connected in $G_b[\CD] \setminus \{u\}$, adding the edge $(v, w)$ does not 
	connect $v$ and $x$ and therefore $u$ is still an articulation point after the 
	insertion of the edge.
	Therefore, all the articulation points that can be removed are in the cycle 
	$v...w,v$. Since they are not articulation points in $G_a[\CD]$, they separate 
	$G_b[\CD]$ in only two components: one with $v$, the other with $w$. Therefore, 
	$v...w,v$ is the only cycle containing $v$ and $w$, and removing any 
	vertex from it make the articulation points of $G_b[\CD]$ be articulations point 
	in $G_a[\CD]$, because they disconnect $v$ and $w$ again.
\end{proof}

Lemma \ref{lem:oneremoval} allows us to focus 
on the following 
problem: find a vertex in $\C$ that is no longer an articulation point in 
$G[\CD]$ after the insertion of the edge $(v, w)$.
To achieve this, the algorithm maintains for each vertex $v \in \C$ the number 
$nc(v)$ of connected component of $G[\D \setminus v]$.
For $v \notin \C$ we set for convenience $nc(v)$ to be the number of connected component 
in $G[\D \setminus v]$ plus $n$. 
This information can be used as follows: when an edge $(v, w)$ is 
added, if for one vertex $u \in \C$ it holds $nc(u)=1$ then $u$ is 
removed from $\C$ (because it is no longer an articulation point). 
To identify such a vertex, the algorithm queries for the minimal value along 
the 
path $v...w$ in $T$: if the minimum value is $1$, the corresponding vertex is 
removed from $\C$. 
This removal makes all the other vertices of the set $\C$ articulation 
points again: by Lemma \ref{lem:oneremoval}, the cycle created by the insertion of $(v, w)$ is broken by  
the deletion of $u$ from $G[\CD]$ . 

Notice that we are only interested in the $nc(v)$ values of the vertices in 
$\C$, as  $nc(v) > n$ for $v \notin \C$. Since we compute a minimum and the 
values relevant are smaller than $n$, this is equivalent to ignoring $v$. 
The advantage of this offset is that when $v$ becomes part of $\C$, it is 
sufficient to decrease its value by $n$ to make it consistent.
We now show how to keep this value up to date after adding or removing an edge. 

\bigskip

\noindent{\bf Maintaining the $nc(v)$ values in a top-tree.}
For this purpose, we use the biconnectivity data structure from 
\cite{Holm2001Poly} (called \emph{top-tree}) on the subgraph $G[\CD]$. 
To avoid cumbersome notation, we pretend that we execute the algorithm on $G$, 
although the underlying graph on which we execute the algorithm is $G[\CD]$. 
We also assume that the number of vertices remains $n$ throughout the 
execution, 
which is simply implemented by removing from $G$ all incident edges from the vertices 
with no incident edges in $G[\CD]$.

We now briefly describe the approach of 
\cite{Holm2001Poly}.
The algorithm maintains a spanning forest $F$ of $G$ and assigns a level 
$\ell(e)$ to each edge $e$ of the graph. Let $G_i$ be the graph composed of $F$ 
and all edges of level at least $i$. The levels are attributed such that the 
following invariant is maintained:

\begin{invariant}
	\label{inv:size-of-bic}
	The maximal number of vertices in a biconnected component of $G_i$ is 
	$\lceil n/2^i\rceil$.
\end{invariant}

\noindent Therefore the algorithm  needs only to consider $\lceil \log_2 
n\rceil$ levels. 
Whenever an edge $(v, w)$ is deleted, one needs to find which vertices in the 
path $v...w$ in $F$ are still biconnected. We use the following notion to describe
the algorithm.

\begin{definition}
	A vertex $u$ is \emph{covered} by a nontree edge $(x, y)$ if it is
	contained in a tree cycle induced by $(x, y)$. 
	We say that a path $v...w$ is covered at level $i$ if every of its node is in a 
	tree cycle induced by an edge at level greater than $i$.
\end{definition}

\noindent
Mark that all the vertices that are covered by a given edge are in the same biconnected 
component.

When a non-tree edge $(v, w)$ is removed, it may affect the 
2-edge connected components along the tree-path $v...w$ in $T$.
To find which vertices are affected, the following algorithm is used in 
\cite{Holm2001Poly}.
It first marks the vertices in $v...w$ 
as no longer covered at level $\ell(v, w)$. Then, it iterates over edges 
$(x, y)$ that could cover $v...w$, i.e., the ones such that the intersection 
between $x...y$ and $v...w$ is not empty, and marks the vertices in this intersection
as covered. This step is explained in the following function, which is 
called for all level $i$ from $\ell(v, w)$ down to 0. $meet(v, w, 
x)$ is the intersection of the tree paths $v...w$, $v...x$ and $x...w$.

\begin{adjustwidth}{0.6cm}{}
	{\bf Recover($v, w, i$).} Set $u := v$, and iterate over the vertices of 
	$v...w$ towards $w$. For each value of $u$, consider each nontree edge $(q, r)$ 
	with $meet(q, v, w) = u$ and such that $u...q$ is covered at level $i$. If it 
	is possible without breaking Invariant 2, increase the level of $(q, r)$ to 
	$i+1$ and mark the edges of $q...r$ covered at level $i+1$. 
	Otherwise, mark them covered at level $i$ and stop. 
	If the phase stopped, start a second symmetric phase with $u = w$ and iterating 
	on $w...v$ towards $v$.
\end{adjustwidth}

\noindent
As shown in \cite{Holm2001Poly}, this is correct and runs in $O(\log^4)$ amortized time.

\begin{figure}[t] 
	\centering
	\includegraphics[width=0.5\textwidth]{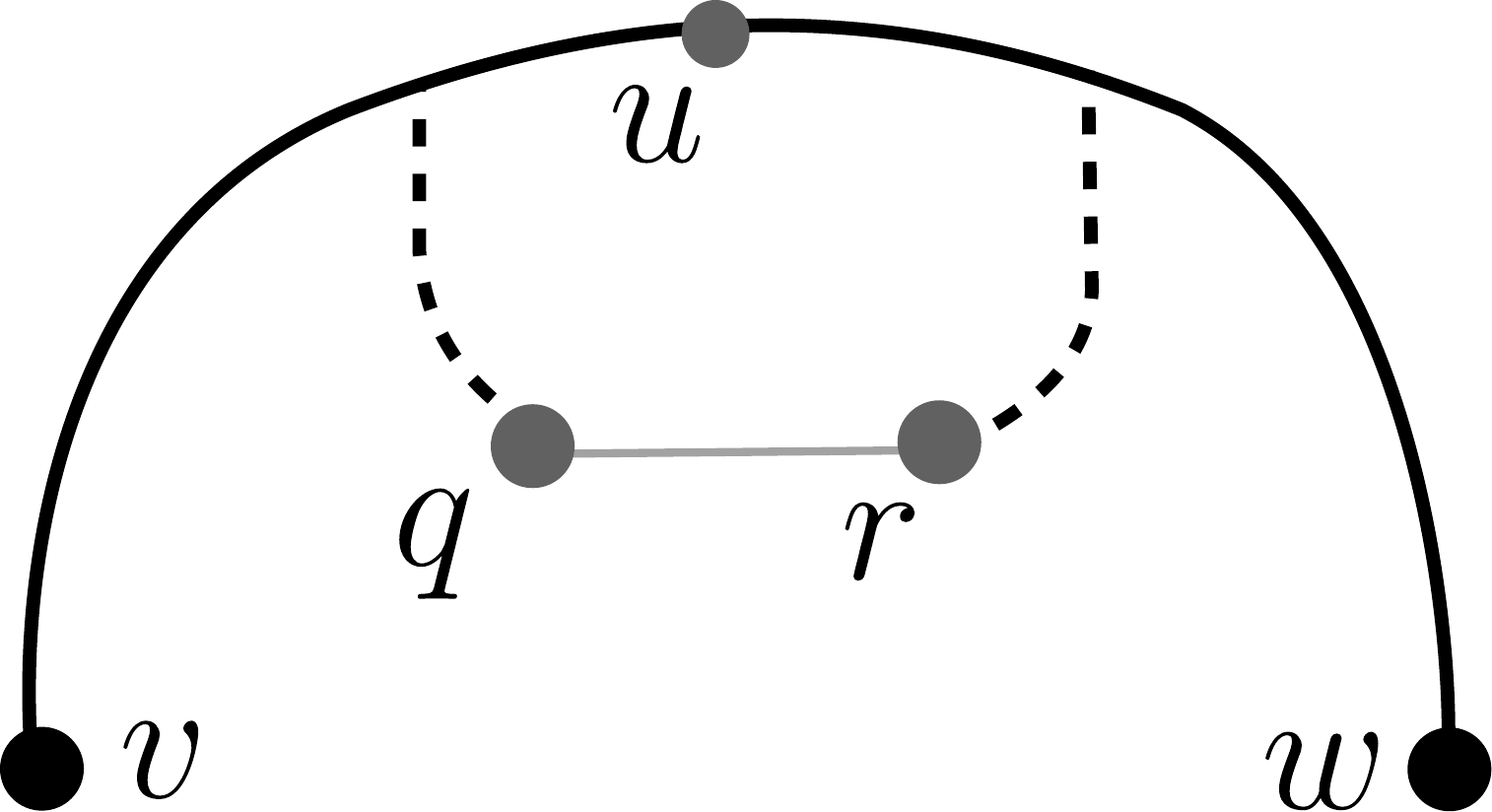}
	\caption{The edge $(q, r)$ covers some node $u$ on the path $v...w$.}
	\label{fig:covering}
\end{figure}

\bigskip
In our case, we are interested in the vertices $u$ whose value $nc(u)$ 
changes. They are exactly those that are still marked as not covered at the end 
of the process. Indeed, if an edge $(q,r)$ covers a vertex $u$ (see Figure~\ref{fig:covering}),
then $v$ and $w$ are still connected in $G[D \setminus u]$, hence the connected
component of $G[D \setminus u]$ do not change. However, if $u$ is not covered by 
any edge, then $v$ and $w$ gets disconnected in $G[D \setminus u]$, thus $nc(u)$
must be updated.
We maintain the $nc( \cdotp)$ values in a 
top-tree, as follows. We call a \emph{segment} a subpath of $v...w$.
The idea is to maintain the non-covered segments and decrease the $nc$ values 
along these at the end of the process.
The top-trees allow us to alter the value of a segment of a path in 
$O(\polylog n)$ time.

\subparagraph*{Computing the list of uncovered segments.}

To find the uncovered segments (in red on Figure~\ref{fig:uncovered}), we sort the covered ones 
and take the complementary.
Let $(q_1, r_1), ..., (q_k, r_k)$ be the nontree edges considered in the execution
of Recover, and let $x_i = LCA(v, q_i)$ and $y_i = LCA(v, r_i)$ (where $LCA(u, v)$
is the lowest common ancestor of $u$ and $v$ in the tree). The covered segments
are exactly the $(x_i, y_i)$. Using lowest common ancestor queries, it is possible
to sort those segments according to the position of $x_i$ along the path $v...w$.
Given the segments in order, it is then possible to determine the uncovered segments
in linear time: they correspond to the complementary of those segments. Answering
a lowest common ancestor query on a dynamic tree can be done in $O(\log n)$ 
(see \cite{sleator1983data}), hence it is possible to sort the covered segments in time
$O(k\log^2 n)$ and to find the uncovered segments with the same complexity.

Since $k$ is the number of edges that move to a higher level during a call to
Recover, and the maximum level is $\log n$, the total complexity of computing
the uncovered segments is at most $\log ^3 n$ per edges. Hence the overall 
complexity is $O(\log^4 n)$, which is the cost of the function Recover.

\bigskip
\begin{figure}[h!] 
	\centering
	\includegraphics[width=0.9\textwidth]{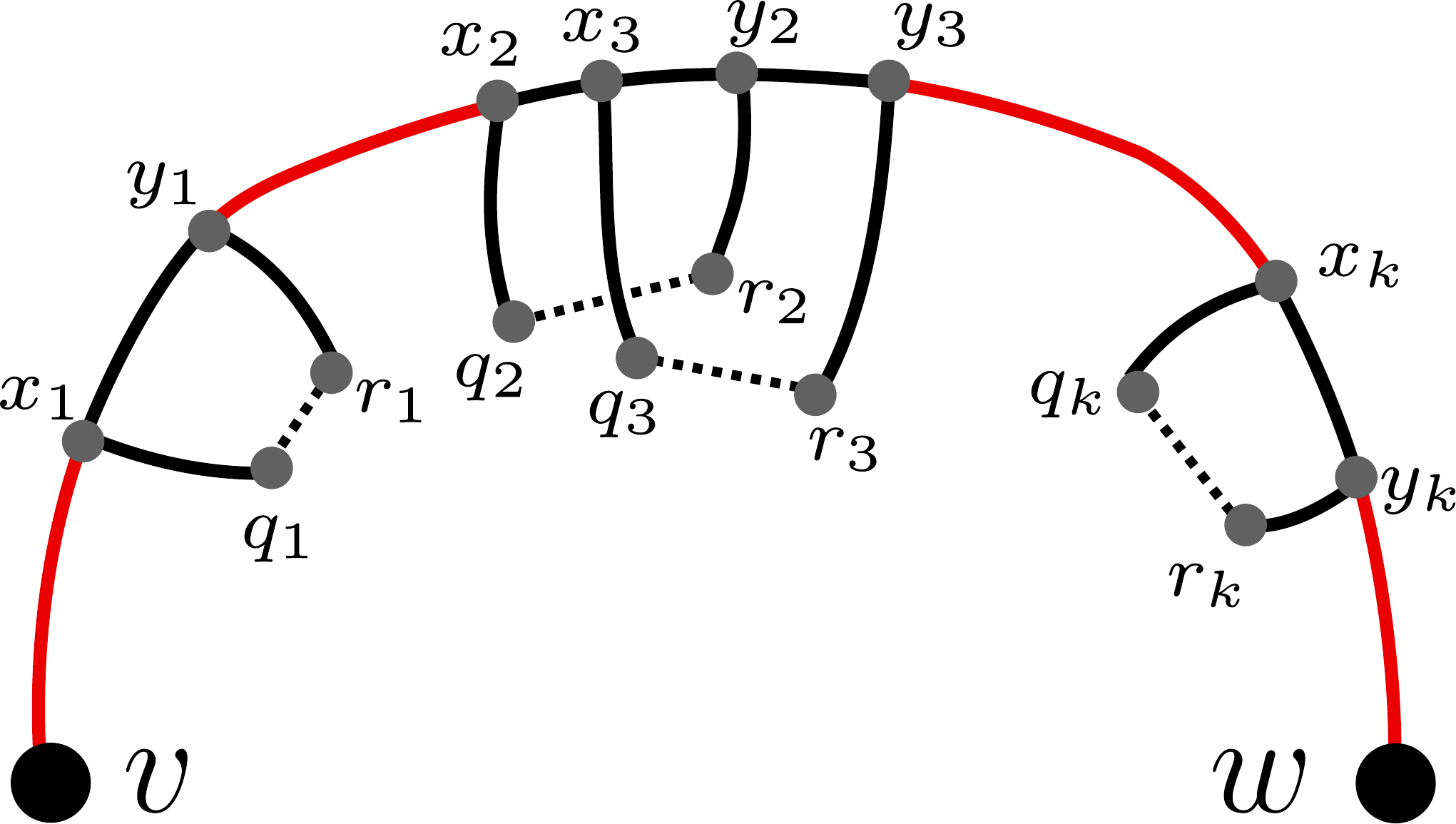}
	\caption{The black segments are covered by edges $(q_i, r_i)$. 
		The red segments are uncovered.}
	\label{fig:uncovered}
\end{figure}

\subparagraph*{\bf Adding an edge.}
To add an edge, two things are required: first decrease some $nc$ value, and 
then query if a vertex has a $nc$ value 1. We have to decrease the $nc$ value 
of a vertex $y$ if and only if its predecessor and its successor along the tree 
path $v...w$ were not connected in $\D \setminus \{y\}$ before the insertion of 
$(v, w)$. 
This turns out to be equivalent to saying that $y$ is not covered: 
thus, the algorithm needs to compute the list of segments along $v...w$ that 
were uncovered before the insertion of $(v, w)$. It then must decrease the $nc$ 
values along these segments, because they become connected. This is analogous to the case of 
an edge deletion: the latter can be 
used the following way. First add the edge $(v, w)$ (and make updates to the 
data structure according to \cite{Holm2001Poly}), then delete it using the 
algorithm from the previous section, with the only difference that, instead of increasing the $nc$ values 
along the uncovered segments, the algorithm decrease them.

It is then easy to find the minimum $nc$ value along the path $v...w$, using 
the top-tree. If this value is 1, we can remove the corresponding vertex from 
$\C$. To remove it, we remove its incident edges one by one, each time updating 
the $nc$ values of the remaining vertices.

\bigskip
The results of this section are summarized in the following lemma.
\begin{lemma}\label{lem:oneedgeupdate}
	After these updates, $\C$ is minimal. Moreover, the algorithm  
	runs in amortized time $\widetilde O(1)$ for a single edge update.
\end{lemma}

A direct corollary of this lemma and Lemma~\ref{lem:simple_complex} is 
Theorem~\ref{thm:connected}.

\begin{corollary}[Theorem~\ref{thm:connected}]
	The whole algorithm to maintain the Connected DS is correct and runs 
	in time $\widetilde O(\Delta)$
\end{corollary}

\begin{proof}
	The correctness follows from Lemma~\ref{lem:oneedgeupdate} and from the correctness of the $\widetilde O(n)$ algorithm.
	As for the running time, the only difference from Lemma~\ref{lem:simple_complex} is
	the search for articulation points: this takes $\widetilde O(1)$ for each edge added
	or removed from $\CD$, and consequently $\widetilde O(\Delta)$ for each node added to 
	or removed from $\CD$. 
	This yields that the algorithm takes $\widetilde O(\Delta)$ amortized time per update.
\end{proof}

\section{A $O(\min(\Delta, \sqrt{m})\,)$ amortized algorithm for {\sc Minimal 
		Dominating Set}}
\label{sec:dom-set}

This section presents a faster algorithm if one is only interested in finding a 
\emph{Minimal} DS. This is a DS in which it is not possible to remove a vertex, 
but it can be arbitrarily big. For instance, in a star, the Minimum DS is only one 
vertex (the center), but its complementary is another minimal DS and has size 
$n-1$. This result highlights the difference between MIS and Minimal DS: the 
best known \emph{deterministic} complexity for MIS is $O(m^{2/3})$, whereas we 
present here a $O(\sqrt m)$ algorithm for Minimal DS.

\subparagraph*{Key idea.} When one needs to add a new vertex to the dominating set 
in order to dominate a vertex $v$, he can choose a vertex with degree 
$O(\sqrt m)$, either $v$ or one of its neighbors (a similar idea
appears in Neiman et al. \cite{neiman2016simple}). We present an 
algorithm with complexity proportional to the degree of the vertex added to the 
DS: this will give a $O(\min(\Delta, \sqrt m))$ algorithm. To analyze the complexity,
we follow an argument similar to the one for CDS. At most one vertex is added to
the DS at every step, even though several can be removed. Therefore we can pay 
for the (future) deletion of a vertex at the time it enters the DS.

For a vertex $v$, $N(v)$ is the set of its neighbors, including $v$. 
Let $\D$ be the dominating set maintained by the algorithm. If $v \in \D$ and 
$u \in N(v)$, we say that $v$ \emph{dominates} $u$.

For each vertex 
$v$, the algorithm keeps this sets up-to-date:
\begin{itemize}
	\item let $N_\D(v)$ be the set of neighbors of $v$ that are 
	in the dominating set $\D$, i.e., $N_\D(v) = \D \cap N(v)$
	\item if $v \in \D$, let $\ID(v)$
	be the set of neighbors of $v$ that are dominated only by $v$, i.e., $\ID(v) = 
	\{u \in N(v) ~~|~~~
	|N_\D(u)| = 1\}$
	
\end{itemize}

Note that $N_\D(v)$ and $\ID(v)$ are useful to check, throughout any sequence 
of 
updates, whether 
a  vertex $v$ must be added 
to or removed from the current dominating set.
In particular, if $N_\D(v) = \emptyset$ then $v$ is not dominated by any 
other vertex, and thus it must be included in the dominating set. On the other 
hand, 
if $\ID(v) = \emptyset$, all the 
neighbors of $v$ ($v$ included) are already dominated by some other vertex, and 
thus $v$ could be removed from the dominating set.

\subsection{The algorithm}
We now show how to maintain a minimal dominating set $\D$ and the sets $N_\D(v)$ 
and 
$\ID(v)$, for each vertex $v$, under arbitrary sequences of edge insertions and 
deletions. We first 
describe two basic primitives, which will be used by our insertion and deletion 
algorithms: adding a vertex to and deleting a 
vertex from a dominating set $\D$.

\subparagraph*{Adding a vertex $v$ to $\D$.}
Following some edge insertion or deletion, it may be necessary to add a vertex 
$v$ to the current dominating set $\D$. In this case, we scan all its neighbors 
$u$ and add $v$ to the sets $N_\D(u)$. If before the update $N_\D(u)$ consisted 
of a single vertex, say $w$, we also have to remove $u$ from the set $\ID(w)$, 
since now $u$ is dominated by both $v$ and $w$. If  $\ID(w)$ becomes empty 
after 
this update, we remove $w$ from $\D$ since it is no longer necessary in the 
dominating set.

\subparagraph*{Removing a vertex $v$ from $\D$.}
When a vertex $v$ is removed from the dominating set, we have to remove $v$ 
from all the sets $N_\D(u)$ such that $u \in N(v)$. If after this update 
$N_\D(u)$ consists of a single vertex, say $w$, we add $u$ to $\ID(w)$.

\subparagraph*{Edge insertion.}
Let $(u, v)$ be an edge to be inserted in the graph. We distinguish three cases 
depending on whether $u$ and $v$ are in the dominating set $\D$ before the 
insertion. If neither of them is in the dominating set (i.e., $u \notin \D $ 
and 
$v \notin \D$), then nothing needs to be done. If both are in the dominating 
set 
(i.e., $u \in \D$ and $v \in \D$), then we start by adding $v$ to the set 
$N_\D(u)$. If $u$ was only necessary to dominate itself, we remove $u$ from 
$\D$. Otherwise, we add $u$ to $N_\D(v)$ and perform the same check on $v$.

If only one of them is in the dominating set (say, $u \notin \D $ and $v \in 
\D$), we have to add 
$v$ to the set $N_\D(u)$. As in the case of adding a vertex to $\D$, this 
may cause the removal of another vertex from the dominating set. 
This can happen only if before the insertion, $N_\D(u) = \{w\}$ for some vertex 
$w$ 
and $\ID(w) = \{u\}$: in other terms, $u$ was dominated only by $w$, and $w$ 
was 
in the dominating set only to dominate $u$. Since after the addition of the 
edge 
$(u, v)$ $u$ is also dominated by $v$, $w$ can be removed from the dominating 
set.

\subparagraph*{Edge deletion.}
Let $(u, v)$ be the edge being deleted from the graph. We distinguish again the 
same three cases as before. 
If $u \notin \D $ and $v \notin \D$, nothing needs to be done.
If both $u \in \D$ and $v \in 
\D$, we just have to remove $u$ (resp. $v$) from the sets $N_\D(u)$ and 
$\ID(u)$ (resp. $N_\D(v)$ and $\ID(v)$).

If only one of them is in the dominating set, say 
$u \notin \D $ and $v \in \D$, then we have to 
remove $v$ from $N_\D(u)$. Now, there are two different subcases:
\begin{itemize}
	\item If $N_\D(u) \neq \{v\}$ before the deletion, then nothing needs to be 
	done.
	\item Otherwise, we have to remove $u$ from $\ID(v)$: if $\ID(v) = \emptyset$ 
	after this operation, then we can safely remove $v$ from $\D$. The algorithm 
	must find a new vertex to dominate $u$: we simply add $u$ to the dominating set.
\end{itemize}

\subsection{Running time}
Adding or removing a vertex $v$ from the dominating set can be done in time
$O(deg(v))$, where $deg(v)$ is the degree of $v$ in the current graph. While 
several vertices can be removed from $\D$ at every step, only one can be added 
(following an edge deletion): the amortized complexity of the algorithm is 
therefore $O(\Delta)$, where $\Delta$ is an upper bound on the degree of the 
nodes.

Nevertheless, it is possible to chose the vertex to be added to the dominating  
set more carefully. When the algorithm must find a new vertex to dominate 
vertex 
$u$, it does the following:

\begin{itemize}
	\item If $deg(u) \leq 2\sqrt m + 1$, the algorithm simply adds $u$ to $\D$.
	\item Otherwise, $deg(u)  > 2 \sqrt m + 1$. The algorithms finds a vertex $w 
	\in 
	N(u)$ with $deg(w) \leq \sqrt{m}\,$ and adds $w$ to $\D$. Note that such a 
	vertex $w$ can be found by simply scanning only  $2\sqrt m + 1$ neighbors of 
	$u$, 
	since (by averaging) at least one of them must have degree smaller than $\sqrt{m}$.
\end{itemize}
In both cases, the insertion takes time $O(\min(\Delta, \sqrt{m})\,)$.

When a vertex $v$ is deleted from the dominating set, its degree can be 
potentially larger than $2\sqrt m$. However, when $v$ was added to the 
dominating 
set its degree must have been $O(\sqrt m\,)$: this implies that many edges were 
added to $v$, and we can amortize the work over those edges. More precisely, 
when a vertex $v$ enters the dominating set, we put a budget $deg(v)$ on it. 
Every time an edge incident to $v$ is added to the graph, we increase by one 
this budget, so that when $v$ has to be removed from $\D$, $v$ has a budget 
larger than $deg(v)$ that can be used for the operation.

\bibliography{bibliography}

\end{document}